\documentclass[letterpaper,11pt]{article}
\usepackage[english]{babel}

\usepackage{amsmath}
\usepackage{amsthm}
\usepackage{enumitem}
\usepackage{amssymb}
\usepackage[margin=2.4cm]{geometry}
\usepackage{todonotes}
\usepackage{graphicx}
\usepackage{bbm}
\usepackage{subcaption}
\usepackage{xcolor}
\usepackage{hyperref}
\usepackage{nicefrac}
\hypersetup{
	colorlinks,
	linkcolor={red!50!black},
	citecolor={blue!50!black},
	urlcolor={blue!80!black}
}
\usepackage[nameinlink]{cleveref}
\usepackage{thm-restate}
\usepackage[linesnumbered,ruled,vlined]{algorithm2e}
\Crefname{algocf}{Algorithm}{Algorithms}
\newtheorem{theorem}{Theorem}
\newtheorem{lemma}[theorem]{Lemma}
\newtheorem{thm}[theorem]{Theorem}
\newtheorem{corollary}[theorem]{Corollary}

\theoremstyle{remark}
\newtheorem{claim}[theorem]{Claim}
\newcommand{\claimproof}[2]{\begin{proof}[Proof of #1]\renewcommand{\qedsymbol}{$\diamond$}#2\end{proof}}

\newcommand{\ma}[1]{{\color{red} #1}}

\usepackage[linesnumbered,ruled,vlined]{algorithm2e}

\usepackage{tikz}
\usetikzlibrary{patterns,patterns.meta,decorations.pathreplacing}
\tikzset{
  picked/.style={
    pattern={Lines[angle=45,distance=2pt]},
    pattern color=gray
  }
}
\tikzset{
  milestone/.style={
    pattern={Hatch[angle=45,distance=1pt]},
    pattern color=gray
  }
}

\Crefname{algocf}{Algorithm}{Algorithms}

\allowdisplaybreaks

\newcommand{\eps}{\varepsilon}
\newcommand{\opt}{{\rm OPT}}

\newcommand{\cost}{{\rm cost}}

\newcommand{\N}{\mathbb N}
\newcommand{\R}{\mathbb R}
\newcommand{\OPT}{\mathrm{OPT}}

\newcommand{\NA}{\mathrm{NA}}
\newcommand{\set}{\mathrm{set}}

\DeclareMathOperator{\E}{E}

\title{Non-Adaptive Evaluation of $k$-of-$n$ Functions:\\ Tight Gap and a Unit-Cost PTAS}
\author{Mads Anker Nielsen\thanks{Department of Mathematics and Computer
    Science, University of Cologne, Cologne, Germany. Email: \texttt{m.nielsen@uni-koeln.de}.} \and Lars
    Rohwedder\thanks{Department of Mathematics and Computer Science,
    University of Southern Denmark, Odense, Denmark. Email: \texttt{rohwedder@imada.sdu.dk}.} \and Kevin
    Schewior\thanks{Department of Mathematics and Computer
    Science, University of Cologne, Cologne, Germany, and Department of Mathematics and Computer Science,
    University of Southern Denmark, Odense, Denmark. Email: \texttt{schewior@cs.uni-koeln.de}.}
}

\begin{document}
\maketitle 

\begin{abstract}
    We consider the Stochastic Boolean Function Evaluation (SBFE) problem in the well-studied case of $k$-of-$n$ functions: There are independent Boolean random variables $x_1,\dots,x_n$ where each variable $i$ has a known probability $p_i$ of taking value $1$, and a known cost $c_i$ that can be paid to find out its value. The value of the function is $1$ iff there are at least $k$ $1$s among the variables. The goal is to efficiently compute a strategy that, at minimum expected cost, tests the variables until the function value is determined. While an elegant polynomial-time exact algorithm is known when tests can be made adaptively, we focus on the non-adaptive variant, for which much less is known.

    First, we show a clean and tight lower bound of $2$ on the adaptivity gap, i.e., the worst-case multiplicative loss in the objective function caused by disallowing adaptivity, of the problem. This improves the tight lower bound of $3/2$ for the unit-cost variant.

    Second, we give a PTAS for computing the best non-adaptive strategy in the unit-cost case, the first PTAS for an SBFE problem. At the core, our scheme establishes a novel notion of two-sided dominance (w.r.t.\ the optimal solution) by guessing so-called milestone tests for a set of carefully chosen buckets of tests. To turn this technique into a polynomial-time algorithm, we use a decomposition approach paired with a random-shift argument.
    
    In fact, our PTAS extends to the class of arbitrary \textit{symmetric} Boolean functions, which are Boolean functions whose value only depends on the number of $1$s among the input variables.
\end{abstract}

\section{Introduction}

The Stochastic Boolean Function Evaluation (SBFE) problem is a fundamental problem in stochastic combinatorial optimization, see e.g.\ the surveys by Ünlüyurt~\cite{unluyurt2004sequential,Unluyurt25} or the more recent works~\cite{deshpande2016approximation,GhugeGN25,hellerstein2024quickly,TanXN25,HarrisNT25,KelesHMMY26,HellersteinPS26}. A function $f:\{0,1\}^n\to\{0,1\}$ is given (typically in compact representation), and the task is to find out $f(x_1,\dots,x_n)$ where $x_1,\dots x_n$ are independent Boolean random variables. For each $i\in[n]$, $p_i\in(0,1)$ is the known probability that $x_i=1$, and the value of $x_i$ can be learned by a policy at known cost $c_i\geq 0$. 

There are two fundamental paradigms for policies: the adaptive one and the non-adaptive one. An adaptive policy may make decisions to test variables depending on the outcomes of previous tests, i.e., it can be viewed as a decision tree. A non-adaptive policy, in contrast, is simply specified by an order of variables to test, which may not be adapted depending on the outcomes of tests. In either case, policies are evaluated by the expected cost paid until the value of $f$ is determined (with probability $1$). While non-adaptive policies have a simple representation, are easy to execute, and are therefore often desirable from a practical point of view, they are generally subobtimal. The adaptivity gap~\cite{DeanGV08,hellerstein22} of a class of functions measures the severity of precisely this suboptimality: the worst-case (in this class of functions) multiplicative gap between the expected cost of the best non-adaptive policy and that of the best adaptive policy. 

An important class of functions are $k$-of-$n$ functions. Such a function is simply given by an integer $k\in[n]$, and the function value is $1$ if and only if there are at least $k$ $1$s among the input variables, i.e., $x_1 + \cdots + x_n \ge k$. SBFE has also been considered for a number of more general classes of functions, e.g., linear threshold functions~\cite{deshpande2016approximation,jiang2020algorithms,GhugeGN25}, symmetric Boolean functions~\cite{gkenosis22stochastic,GhugeGN25,liu20226approximation,plank2024simple}, and voting functions~\cite{hellerstein2024quickly}. For many of these classes, polynomial-time approximation algorithms to compute the best non-adaptive or adaptive policy have been proposed.

Apart from it occurring as a special case of many SBFE problems studied in the literature, another reason for the popularity of $k$-of-$n$ functions may be the elegance of the optimal policy~\cite{SalloumBreuer84,BenDov81}: Conditional on function value $0$ and $1$, it is optimal to test in increasing order of $c_i/(1-p_i)$ and $c_i/p_i$ ratios, respectively, to find a certificate at minimum expected cost. Since these policies need to test at least $n-k+1$ and $k$ tests, respectively, to find a certificate, there is, by the pigeon-hole principle, some test occurring in both these prefixes, which can safely be tested even in the original (unconditional) case.

When non-adaptive policies come into play, however, much less is known. First, while it is known that, in the unit-cost case the adaptivity gap is exactly $3/2$~\cite{grammel2022,plank2024simple}, no stronger lower bound for the general case is known, with the upper bound only being $2$~\cite{gkenosis2018stochastic}. Second, the polynomial-time approximation algorithms implied by the upper-bound proofs on the adaptivity gaps are those with the best known approximation ratios, i.e., $3/2$ in the unit-cost case and $2$ in the general case. In this paper, we make significant progress on both questions.

Regarding the second question, we also make progress on the aforementioned symmetric Boolean functions. These are functions whose value only depends on the number of $1$s among the input variables.
For arbitrary costs, the state-of-the-art approximation guarantee of a polynomial-time algorithm is $5.829$~\cite{plank2024simple}.
Again, this algorithm computes a non-adaptive policy but the guarantee compares to the best adaptive policy.
In the unit-cost case, the guarantee can be improved to $2$~\cite{grammel2022}.
Another special case of this class of functions is the unanimous-vote function, for which exact polynomial-time adaptive~\cite{gkenosis2018stochastic} and, in the unit-cost case, non-adaptive~\cite{KelesHMMY26} algorithms are known.

\subsection{Our Contribution}

Our first result is the following.

\begin{restatable}{theorem}{gapmain}
\label{thm:gapmain}
The adaptivity gap of SBFE on $k$-of-$n$ functions is exactly $2$.
\end{restatable}

This settles an open question that had been known within the community and is explicitly stated in~\cite{plank2024simple}. Since an upper bound of $2$ was known~\cite{gkenosis2018stochastic}, our contribution  is showing a matching lower bound. Our construction is perhaps surprisingly simple.

To give an overview, let us first recall the class of instances for the unit-cost version which leads to a lower bound of 3/2~\cite{plank2024simple}. Here, $n=2t+1$ for some integer $t$, and $k=t+1$. There are $t$ \emph{$1$-variables} taking value $1$ and $t$ \emph{$0$-variables} taking value $0$. Note that we do not allow probabilities of precisely $0$ or $1$ in our model, but we do in this section for the simplicity of exposition. Such variables still have to be tested in order to use them towards certifying the function value. In addition, there is a single \emph{pivotal} variable with probability $1/2$, whose value thus determines the function value. An adaptive policy can simply test the pivotal variable, having, say, value $i$, and then test all the $i$-variables, at expected cost of $t+1$. A non-adaptive policy, on the other hand, can only ``guess'' the outcome of the pivotal variable, resulting in an expected cost of $\nicefrac32\cdot t+1$. Then taking the limit $t\to\infty$ yields the result.

Note that, in the arbitrary-cost case, we can simply assign the pivotal variable a cost of $0$, so that $t=1$ would actually suffice to obtain a bound of $3/2$. The main idea of our construction is to have several, say, $m$, pivotal variables, all with probability $1/2$ and cost $0$. It is easy to see that, with $t=1$, the resulting ratio is still $3/2$. The correct regime turns out to be the one in which $t$ is large but $m$ is much larger than $t$. In that case, if one is not done after performing the pivotal tests (which are free), the function value is determined to be $0$ or $1$ with probability $1/2$ each, and the number of tests required to prove that is (in the limit) uniformly distributed between $1$ and $t$. These quantities are revealed to an adaptive policy but not to a non-adaptive policy. Hence, in the limit, the cost of an adaptive policy is $\nicefrac 1t\cdot \sum_{i=1}^t i=(t+1)/2$ and for a non-adaptive policy, which is done at any step during the remaining $2t$ steps with equal probability, $\nicefrac 1{(2t)}\cdot \sum_{i=1}^{2t} i=(2t+1)/2$. The result then follows with $t\to\infty$.

Our second result is not structural but algorithmic.

\begin{restatable}{theorem}{ptasmain}
\label{thm:ptasmain}
There is a PTAS for computing the optimal non-adaptive policy for evaluating symmetric Boolean functions in the unit-cost case.
\end{restatable}

Notably, this is the first PTAS for an SBFE problem. 
The PTAS relies on the idea of carefully enumerating a polynomial number of policies.
We denote the optimal policy by $\pi^\star$.
Ideally, we would like to guarantee that among the enumerated policies there is one policy $\pi$ which satisfies for all $i$ that
	\begin{equation*}
		\Pr[\cost(\pi) \ge i]\le \Pr[(1 + \eps) \cost(\pi^\star) \ge i].
	\end{equation*}
This would be enough to show that $\E[\cost(\pi)]\le (1 + \eps) \E[\cost(\pi^\star)]$.

It would be possible, albeit non-trivial, to get this property in quasi-polynomial time: View a policy as a sequence of buckets (containing tests) of exponentially increasing size, so that it essentially does not matter (in terms of a $(1+\eps)$-approximation) where one finishes in the bucket. Then, starting from the left, for each bucket, apply our core idea: In the corresponding bucket of $\pi^\star$, sort all variables by their probability. In this order, we guess $1/\eps$ equally spaced tests. Our main structural result is to show that if one has correctly guessed these equally spaced tests, then only using a few extra tests in each bucket (as compared to the corresponding bucket in OPT) we can achieve a strong domination property that we will describe below.

Consider any prefix $P$ of buckets of the algorithm and the corresponding prefix of buckets $P^\star$ of the optimum. Then our domination property says the following: There is an injection from the tests in $P^\star$ to those in $P$ such that, for those in $P$ the probability is at least that of the corresponding test in $P^\star$; moreover, the there also exists an injection for the other direction, i.e., one where the probability is at most as high. Crucially, this domination property allows us to argue that the probability that our policy determines $f$ by performing the tests in $P$ is at least as large as the probability of achieving the same goal by performing the tests in $P^\star$.

The reason that the described scheme is only a QPTAS is that the number of buckets is logarithmic and the algorithm guesses all combinations of possible milestones, across all buckets. To obtain a PTAS, we decompose our instance into subinstances that do not interact with each other. The subinstances are instances of a weaker variant, in which the above inequality only has to be fulfilled for a bounded range of values $i$, so that the number of buckets is also bounded. We use a random-shift argument to show that such a decomposition actually exists. Specifically, we show that we can find an offset of ranges at exponentially increasing times at which one can afford to ``reset'', i.e., create a new subinstance that does not interact with any of the previous ones.

\subsection{Further Related Work}

As we focus exclusively on $k$-of-$n$ and symmetric functions, we only briefly review the SBFE literature for function classes other than those mentioned above. One example are functions that take value $1$ iff there is an $s$-$t$ path in a given graph, whose edges exist iff some input random variable is $1$. This problem is NP-hard~\cite{fu2017determining,Guo0N0N24} due to connections to the $s$-$t$ reliability problem. Read-once formulas, which are equivalent to the same functions for series--parallel graphs, have also been extensively studied (see~\cite{unluyurt2004sequential} for an overview), but optimal algorithms have only been obtained in special cases; in particular, no constant-factor approximations are known in the general case.

We list a few standard techniques for obtaining approximation algorithms for SBFE problems, all of which inherently incur constant-factor multiplicative losses that are not useful towards designing a PTAS: round-robin approaches (e.g., \cite{kaplanMansour-Stoc05,allen2017evaluation,plank2024simple}), reducing to a submodular cover problem (e.g., \cite{deshpande2016approximation,Hellerstein2021tight,cuiNagarajan-SOSA23}), and a round-based approach (e.g., \cite{ImNZ16,EneNS17,GhugeGN25}).

While there have previously not been PTASs for SBFE problems, a PTAS has been shown to exist~\cite{SegevS22} for the batched variant~\cite{DaldalOSSU17}, a variant where tests can be batched together and there is an additional fixed cost for each batch, in the case of a conjunction ($1$-of-$n$ function). Other examples of stochastic combinatorial-optimization problems for which PTASs have been developed are stochastic probing and prophet problems~\cite{Segev021},  the Pandora's Box problem with non-obligatory inspection~\cite{BeyhaghiC23}, and stochastic scheduling problems~\cite{BuchemERSW24}, but it is unclear whether these techniques can be used for SBFE problems. We hope that our work will spark follow-up work on PTASs for SBFE problems.

Finally, we remark that NP-hardness for computing optimal non-adaptive policies for symmetric Boolean functions is \emph{not} known. In many other cases (e.g., voting functions), this state is the same, and approximation algorithms are developed as hardness is conjectured.
In some cases NP-hardness is known, but it requires a more complicated complicated structure of the corresponding deterministic problem (e.g.,~\cite{deshpande2016approximation,FuFXPWL17}); arguably, we lack techniques for establishing NP-hardness of SBFE problems that stems from the stochasticity of the problems.

\section{Preliminaries}

Throughout the paper, we use $\N=\{1,2,\dots\}$ and $\N_0=\{0,1,2,\dots\}$. For $m\in\mathbb{N}$, we use $[m]$ as a shorthand for $\{1,\dots,m\}$.

We call a function $f: \{0,1\}^n \rightarrow \{0,1\}$ a Boolean function in
$n$ variables. A \textit{state} $s$ is a vector $s \in
\{0,1,\bullet\}^n$. If $s$ is a state and $x \in \{0,1\}^n$, then we say that $x$ follows from
$s$ if $s_i \in \{x_i,\bullet\}$ for all $i \in [n]$. We say that $f$ is \textit{determined} by a state $s$ if $f$ takes the same value for all $x \in \{0,1\}$ following from $s$.

In the SBFE problem, we are given a Boolean function $f$ in $n$ variables, a cost vector $c \in \R^n$ and a probability vector $p \in (0,1)^n$. We assume w.l.o.g.\ that $p_1\leq p_2\leq \dots \leq p_n$. The values of the variables $x_1,x_2,\dots,x_n$ are initially unknown, and at each step, we must select a variable $x_i$ to test, upon which its value is revealed. We represent the current knowledge as a state vector $s \in \{0,1,\bullet\}^n$, where $s_i$ is the currently known value of $x_i$ or $\bullet$ if $x_i$ is unknown. We must continue testing variables exactly until $f$ is determined by the current state $s$.

Formally, we can represent a policy as function $\pi: \{0,1,\bullet\}^n \setminus \{0,1\}^n \rightarrow [n]$ where $x_{\pi(s)}$ is the variable tested when currently in state $s$, with the requirement that $s_{\pi(s)} = \bullet$ for all states $s \notin \{0,1\}^n$ (i.e., only variables with unknown value can be tested). If $\pi$ tests exactly the set of variables with indices $S \subseteq \{1,2,\dots,n\}$ (a random variable) before determining $f$, then the cost of $\pi$ on input $x$ with respect to cost vector $c$ (a random variable) is denoted $\cost_{c}(\pi,x)$ and formally defined defined as \[
    \cost_{c}(\pi,x) = \sum_{i \in S} c_i.
\]
The expected cost $\E_p[\cost_c(\pi)]$ of $\pi$ with respect to probability vector $p$ and cost vector $c$ is the expected value of $\cost_c(\pi)$ with respect to the probability distribution $\Pr$ given by \[
\mathrm{Pr}[x] = \prod_{i \in [n]: x_i = 1}p_i\prod_{i \in [n]: x_i = 0} (1-p_i).
\]
for all $x \in \{0,1\}^n$. We omit the subscripts $p$ and $c$ from $\E_p$ and, respectively, $\cost_c$ when they are clear from context.

A policy $\pi$ is \textit{non-adaptive} if $\pi(s)$ only depends on the number of unknown variables in state~$s$. We represent a non-adaptive policy simply as the fixed order (permutation) $\sigma$ of $[n]$ such that $x_{\sigma(i)}$ is the variable tested in the $i$-th step.

A \emph{partial} non-adaptive policy $\pi$ is a non-adaptive policy that stops early. It can be represented by a permutation $\sigma$ of a subset $\set(\sigma)$ of $[n]$ where again $x_{\sigma(i)}$ is the variable tested in the $i$-th step. 
If $\pi$ determines the value of $f$, its cost is defined in the same way as for non-partial policies. If $\pi$ does not determine the value of $f$, its cost is $n$.

Given two partial non-adaptive policies $\pi_1,\pi_2$, we denote by $\pi_1\circ\pi_2$ the (possibly partial) non-adaptive policy that first tests in order of $\pi_1$ and then in order of $\pi_2$, skipping any test that has been conducted by $\pi_1$ already.

The \textit{optimal policy} $\opt(f,c,p)$ (optimal non-adaptive policy
$\opt_{\NA}(f,c,p)$) of a function $f$ with respect to cost vector $c$ and
probability vector $p$ is the policy (non-adaptive policy) $\pi$ for $f$
which minimizes $\E_p[\cost_{c}(\pi)]$.

We define an instance of the SBFE problem as a triple $I = (f,p,c)$ where $f$ is a Boolean function, $p$ if a probability vector, and $c$ is a cost vector. For the purposes of determining the encoding length of an instance in the case of $k$-of-$n$ functions, $f$ is simply given by the integer $k$. A symmetric function is given through thresholds $0 = t_1 < t_2 <
\dots < t_T < t_{T+1} = n+1$ such that the function values changes at the thresholds, i.e., $f(x) \neq f(y)$ if $t_j \leq
\sum_{i=1}^n x_i < t_{j+1}$ and $t_{j+1} \leq
\sum_{i=1}^n y_i < t_{j+2}$, for $j \in [T-1]$.

Finally, let $\mathcal{I}$ be a class of instances $(f,p,c)$ of the SBFE problem. (In this paper, we will only consider the case where $f$ is a symmetric function.) The \textit{adaptivity gap} is defined as \[
    \sup_{(f,p,c) \in \mathcal{I}}
    \frac{\E_p[\cost_c(\opt_{\NA}(f,c,p))]}{\E_p[\cost_{c}(\opt(f,c,p))]}.
\]

\section{Tight Lower Bound on the Adaptivity Gap}
\label{sec:gap}

In this section, we show \Cref{thm:gapmain}, which we restate here for convenience.

\gapmain*

Recall that an upper bound of $2$ is known~\cite{gkenosis2018stochastic}, so we will proceed by proving a tight lower bound.
Our family of lower-bound instances is defined as follows. For positive integers $m,t$ and $\eps \in (0,1)$ (we only ever pick
$\eps$ close to $0$), define $L_{m,t,\eps} = (f,c,p)$ where
$f$ is the $k$-of-$n$ function with $k=m+t$ and $n=2m+2t$ and
\[
    (c_i,p_i) = \begin{cases}
        (1,\eps) & 1 \leq i \leq t\\    
        (0,1/2) & t < i \leq 2m+t \\
        (1,1-\eps) & 2m+t < i \leq 2m+2t
    \end{cases}
\]
for $i \in [n]$. 
To show the theorem, we will show that
\[
\lim_{t \rightarrow \infty}\lim_{m \rightarrow \infty}\lim_{\eps \rightarrow 0} \frac{\E[\cost(\OPT_{\NA}(L_{m,t,\eps}))]}{\E[\cost(\OPT(L_{m,t,\eps}))]} = 2.
\]
Thus, one should think of $\eps$ as being vanishingly small, of $t$ as being large, and of $m$ as being much larger than $t$.

We refer to the $2m$ variables of $L_{m,t,\eps}$ with $c_i = 0$ as \textit{free} variables and the remaining $2t$ variables as \textit{paid} variables. Among the paid variables, we refer to those with $p_i = 1-\eps$ as $1$-variables and those with $p_i = \eps$ as $0$-variables.\footnote{Note that, for an alternative proof, we could assume that the $i$-variables take value $i$ with probability $1$, for $i\in\{0,1\}$, and then use a continuity argument (or allow probabilities $0$ or $1$ in our model anyway).} We denote by $X$ the random variable such that $X(x) = |\{t < i \leq 2m+t \mid x_i = 1\}|$, i.e., $X$ is the number of $1$s among the free variables.

We assume that an optimal policy for $L_{m,t,\eps}$ tests all free variables before testing any other variable. We call such a policy \textit{economical}. This assumption is clearly without loss of generality since, if some policy is not economical, we can make it economical by moving all free variables to the front without increasing the cost of the policy for any $x \in \{0,1\}^n$.

To get an intuition, one can think of the $1$-variables and $0$-variables as always taking value $0$ and $1$ respectively, which is true in the limit $\eps \rightarrow 0$. Nevertheless, the policies may need to test them and,
in particular, pay their cost, until the function is determined.
This is where an adaptive policy has an advantage over a non-adaptive policy: after the free variables have been tested, an adaptive policy
can behave optimally, since it already knows the function value. The
non-adaptive policy, on the other hand, cannot use the outcome of the
free tests and therefore has to hedge against both possible function values, which leads to the ratio of $2$ as we will see in the remainder.

Since $m$ is much larger than $t$, the event that the function value is not
determined after performing the free tests (formally, $X-m\in[-t,t-1]$) has
low probability. Nevertheless, we observe in the following lemma that this
is the only event in which the behavior of an economical policy matters.

\begin{lemma}
\label{lem:ratio}
    Let $m$ and $t$ be any two integers with $m > t$, let $\eps \in (0,1)$ be arbitrary, 
    and let $X$ be as above. 
    For any pair of economical policies $\pi$
    and $\pi'$ we have \[
        \frac{\E[\cost(\pi)]}{\E[\cost(\pi')]} =
        \frac{\E[\cost(\pi) \mid X-m \in
            [-t,t-1]]}{\E[\cost(\pi') \mid
            X-m \in [-t,t-1]]}.
    \]
\end{lemma}

\begin{proof}
    Recall that there are $n = 2m+2t$ variables in the instance $L_{m,t,\eps}$ and the function value is determined when we find at least $k = m+t$ $1$s or at least $n-k+1$ $0$s.
    Let $F$ be the event that $X-m \in [-t,t-1]$.
    By the law of total expectation
    \begin{equation}\label{eq:cost}
        \frac{\E[\cost(\pi)]}{\E[\cost(\pi')]} =
        \frac{\E[\cost(\pi) \mid
            \overline{F}]\Pr[\overline{F}]+\E[\cost(\pi)
            \mid F]\Pr[F]}{\E[\cost(\pi') \mid
            \overline{F}]\Pr[\overline{F}]+\E[\cost(\pi')\mid F]\Pr[F]}.
    \end{equation}
    Suppose $\overline{F}$ occurs. Then either $X \geq m+t$ or $X \leq m-t-1$. In the first case, the number of $1$s in $x$ is at least $X \geq m+t=k$. In the latter case, the number of $0$s in $x$ is at least $2m-X \geq 2m-(m-t-1) = (2m+t)-m+1 = n-m-t+1 = n-k+1$. In both cases, the function value is determined by any economical policy at cost $0$. Thus, $\E[\cost(\pi)
    \mid \overline{F}] = \E[\cost(\pi') \mid \overline{F}] = 0$ and
    therefore \Cref{eq:cost} simplifies to
    \[
        \frac{\E[\cost(\pi)]}{\E[\cost(\pi')]} =
        \frac{\E[\cost(\pi) \mid F]}{\E[\cost(\pi') \mid
            F]},
    \]
    as desired.
\end{proof}

Next, we show a technical lemma that states that, in the limit case, $X-m$ takes any integer value in the interval $[-t,t-1]$ with the same probability. Intuitively, this follows from Lipschitzness of the normal distribution. 

\begin{lemma}
    \label{lem:binomial}
    For any $a\in\mathbb{N}$, let $X_{2a}$ be a random variable drawn from a binomial distribution with parameters $2a$ (number of trials) and $1/2$ (success probability). Then, for any $c\in\mathbb{N}$ and integer $i \in [-c,c-1]$,  we have
    \[
        \lim_{a \rightarrow \infty} \Pr[X_{2a}-a = i \mid
        X_{2a}-a \in [-c,c-1]] = \frac{1}{2c}.
    \]
\end{lemma}

\begin{proof}
    Fix an integer $c$ and $i \in [-c,c-1]$. Now, for any positive integer
    $a \geq c$
    \begin{align}\label{eq:recip} 
        \nonumber\Pr[X_{2a}-a = i \mid X_{2a}-a \in [-c,c-1]] &= \frac{\Pr[X_{2a}-a
        = i]}{\Pr[X_{2a}-a \in [-c,c-1]]} \\
                     &= \frac{\binom{2a}{a+i}}{\sum_{j=-c}^{c-1} \binom{2a}{a+j}}.
    \end{align}
    We show that the reciprocal of \Cref{eq:recip} converges to $2c$
    as $a \rightarrow \infty$ and the lemma follows by the quotient law for limits. We see that
    \begin{equation}
        \label{eq:sum}
        \frac{\sum_{j=-c}^{c-1} \binom{2a}{a+j}}{\binom{2a}{a+i}} =
        \sum_{j=-c}^{c-1} \frac{\binom{2a}{a+j}}{\binom{2a}{a+i}} =
        \sum_{j=-c}^{c-1} \frac{(a+i)!(a-i)!}{(a+j)!(a-j)!}.
    \end{equation}

    Now, we show that every term in \Cref{eq:sum} converges to $1$ as $a
    \rightarrow \infty$ and the lemma follows from the sum law for limits.
    Consider a fixed integer $j \in [-c,c-1]$ and the expression
    \begin{equation}
        \label{eq:frac}
        \frac{(a+i)!(a-i)!}{(a+j)!(a-j)!}. 
    \end{equation}
    Both the numerator and denominator contain exactly $2a$ terms and $(a-c)!$ occurs twice as a factor in both. Canceling out the factors of $(a-c)!$, we are left with $2c$ factors in both the numerator and denominator. Pairing these factors arbitrarily, we are left with the product of $2$ fractions of the form $(a+c_1)/(a+c_2)$ where $c_1,c_2 \in
    [-c,c-1]$, all of which converge to $1$ as $a \rightarrow \infty$. Thus, by the product law for limits, \Cref{eq:frac} converges to $1$ as $a \rightarrow \infty$ for any fixed $j \in [-c,c-1]$. This completes the
    proof of the lemma.
\end{proof}

To prove \Cref{thm:gapmain}, by the previous two lemmata, we may focus on the setting where $X-m$ is uniformly distributed in $[-t,t-1]$. We analyze the performance of an adaptive policy and that of a non-adaptive policy.

\begin{proof}[Proof of \Cref{thm:gapmain}]
    Let $m,t \in \N$ and $\eps \in (0,1)$ be arbitrary and let $(f,c,p) = L_{m,t,\eps}$. Let $X$ be as above. 
    Since $p_i = 1/2$ for the $2m$ variables $x_i$ with $t < i < 2m+t$ contributing to $X$ and $t$ is a constant independent of $m$, \Cref{lem:binomial} implies \begin{equation}
        \label{eq:xinrange}
    \lim_{m \rightarrow \infty} \E[X-m = i \mid X-m \in [-t,t-1]] = \frac{1}{2t}
    \end{equation}
    for any integer $i \in [-t,t-1]$.

    Now, let $\pi$ be the adaptive policy which on input $x$ evaluates the free variables and then, if it has not yet determined $f$, evaluates all $1$-variables first if $X(x)-m\geq 0$ and otherwise evaluates all $0$-variables first. Also, let $\pi_{\NA}$ be any economical non-adaptive policy. As $\pi_{\NA}$ is non-adaptive and economical, the last $2t$ tests of $\pi_{\NA}$ are the paid variables of $L_{m,t,\eps}$ in some fixed order $\sigma(1),\sigma(2),\dots,\sigma(2t)$.  
    For $j \in [t]$, let $n_{1}(j)$ be minimal such that there are exactly $j$ $1$-variables  in $\{x_{\sigma(1)},x_{\sigma(2)},\dots,x_{\sigma(n_1(j))}\}$ and let $n_{0}(j)$ be minimal such that there are exactly $j$ $0$-variables $\{x_{\sigma(1)},x_{\sigma(2)},\dots,x_{\sigma(n_0(j))}\}$).
    
    \begin{claim}
        \label{claim:nonadaptivestop}
        Let $x \in \{0,1\}$ be such that $X(x)-m \in [-t,t-1]$ and $x_i = 1$ for all $1$-variables $x_i$ and $x_j = 0$ for all $0$-variables $x_j$. Then \begin{align*}
        \cost(\pi,x) &= \begin{cases}
            k-X(x) & \text{if } X(x)-m \geq 0 \\
            (n-k+1)-(2m-X(x)) & \text{if } X(x)-m < 0
        \end{cases}\ \text{and} \\
        \cost(\pi_{\NA},x) &= \begin{cases}
            n_1(k-X(x)) & \text{if } X(x)-m \geq 0 \\
            n_0((n-k+1)-(2m-X(x))) & \text{if } X(x)-m < 0.
        \end{cases}
    \end{align*}
    \end{claim}

    \claimproof{\Cref{claim:nonadaptivestop}}{        
        Suppose $X-m \geq 0$. Since all $t$ $1$-variables evaluate to $1$, there are $X(x)+t \geq m+t = k$ $1$s in $x$ and so $f(x) = 1$.
        Thus, any economical policy has determined the value of $f$ exactly when it has evaluated the ($k-X(x)$)-th $1$-variable. Hence, by the definition of $\pi$, $\cost(\pi,x) = k-X(x)$, and by the definition of $n_1$, $\cost(\pi_{\NA},x) = n_1(k-X(x))$.

        Suppose that $X(x)-m < 0$. Since all $t$ $0$-variables evaluate to $0$, there are $2m-X(x)+t \geq n-k+1$ $0$s in $x$ and so $f(x) = 0$. Thus, any economical policy has determined the value of $f$ exactly when it has evaluated the $((n-k+1)-(2m-X(x)))$-th $0$-variable. Hence, by the definition of $\pi$, $\cost(\pi,x) = (n-k+1)-(2m-X(x))$, and by the definition of $n_0$, $\cost(\pi_{\NA},x) = n_0((n-k+1)-(2m-X(x)))$.
    }
    
    \smallskip

    Note that the condition on $x$ from \Cref{claim:nonadaptivestop} (all $i$-variable take value $i$ for $i=0,1$) holds with probability approaching $1$ as $\eps$ approaches $0$. Using this observation and \Cref{claim:nonadaptivestop}, we will show that \[
        \lim_{m \rightarrow \infty}\lim_{\eps \rightarrow 0} \frac{\E[\cost(\pi_{\NA})]}{\E[\cost(\pi)]} = \frac{2t+1}{t+1},
    \]
    from which the theorem follows by taking the limit as $t$ approaches infinity.
    First, by \Cref{lem:ratio}, and since both expectations are between $0$ and $2t$, the above ratio is equal to 
    \[
         \frac{\lim_{m \rightarrow \infty}\lim_{\eps \rightarrow 0}\E[\cost(\pi_{\NA}) \mid X-m \in [-t,t-1]]}{\lim_{m \rightarrow \infty}\lim_{\eps \rightarrow 0}\E[\cost(\pi) \mid X-m \in [-t,t-1]]}.
    \]

    We first analyze the cost of $\pi$:
    \begin{align*}
        \lim_{m \rightarrow \infty}\lim_{\eps \rightarrow 0}&\E[\cost(\pi) \mid X \in [-t,t-1]] \\
        &= \lim_{m\rightarrow \infty}\lim_{\eps \rightarrow 0}\sum_{i=-t}^{t-1} \mathrm{Pr}[X-m = i \mid X-m \in [-t,t-1]]\cdot \E[\cost(\pi) \mid X-m = i] \\
        &= \sum_{i=-t}^{t-1} \lim_{m\rightarrow \infty}\mathrm{Pr}[X-m = i \mid X-m \in [-t,t-1]]\cdot \lim_{\eps \rightarrow 0}\E[\cost(\pi) \mid X-m = i] \\
        &= \frac{1}{2t}\left(\sum_{i=-t}^{-1} ((n-k+1)-(2m-(i+m))) + \sum_{i=0}^{t-1} (k-(i+m)) \right) \\
        &= \frac{1}{2t}\left(\sum_{i=-t}^{-1} (t+i+1) + \sum_{i=0}^{t-1} (t-i) \right) \\
        &= \frac{1}{2t}\left(\sum_{i=1}^{t} i + \sum_{i=1}^{t} i \right)
        = \frac{1}{t}\sum_{i=1}^{t} i
        = \frac{t(t+1)}{2t}
        = \frac{t+1}{2},
    \end{align*}
    where the second equality is a consequence of the sum and product laws for limits, the third follows from \Cref{eq:xinrange} and \Cref{claim:nonadaptivestop}.

    Using an argument identical to the case of $\pi$ in the first step, we see that for $\pi_{\NA}$,
    \begin{align*}
    \lim_{m \rightarrow \infty}\lim_{\eps \rightarrow 0}&\E[\cost(\pi_{\NA}) \mid X \in [-t,t-1]] \\
        &= \frac{1}{2t}\left(\sum_{i=-t}^{-1} n_{0}((n-k+1)-(2m-(i+m)) + \sum_{i=0}^{t-1} n_{1}(k-(i+m)) \right) \\
        &= \frac{1}{2t}\left(\sum_{i=-t}^{-1} n_{0}(t+i+1) + \sum_{i=0}^{t-1} n_{1}(t-i) \right) \\
        &= \frac{1}{2t}\left(\sum_{i=1}^{t} n_{0}(i) + \sum_{i=1}^{t} n_{1}(i) \right)
        = \frac{1}{2t}\sum_{i=1}^{2t} i 
        = \frac{2t(2t+1)}{2\cdot 2t}
        = \frac{2t+1}{2},
    \end{align*}
    where, in the fourth equality, we used that \[
        \{n_1(i) \mid i \in [t]\} \cup \{n_0(i) \mid i \in [t]\} = [2t]
    \]
    since $\sigma(j)$ is either a $1$-variable or a $0$-variable for all $j \in [2t]$.

    Thus, we have \[
    \frac{\lim_{m \rightarrow \infty}\lim_{\eps \rightarrow 0}\E[\cost(\pi_{\NA}) \mid X-m \in [-t,t-1]]}{\lim_{m \rightarrow \infty}\lim_{\eps \rightarrow 0}\E[\cost(\pi) \mid X-m \in [-t,t-1]]} = \frac{2t+1}{t+1},
    \]
    which is what we wanted to show.
\end{proof}

\section{PTAS for Symmetric Functions in the Unit-Cost Case}
\label{sec:ptas}

In this section, we show \Cref{thm:ptasmain}, which we restate here for convenience.

\ptasmain*

In~\Cref{subsec:reduction}, we first reduce the task to get a PTAS to computing solution with certain properties for a ``bounded'' variant of the problem. Then, in \Cref{subsec:bounded}, we solve this variant of the problem.

Throughout the section we let $\pi^\star:=\opt_{\NA}(f,c,p)$ the optimal non-adaptive policy. We also assume without loss of generality that $1/\eps\in \N$ and use $O_\eps(\cdot)$ to suppress dependencies on $\eps$ in $O(\cdot)$ notation.

\subsection{Reduction to the Bounded Variant}
\label{subsec:reduction}

In this subsection we will prove that the following lemma suffices to obtain a PTAS for our problem. The proof of the lemma is then given in the subsection after.
\begin{lemma}\label{lem:pq}
	Given $\eps > 0$ and $a, a'\in [n]$ with $a < a'$, there is an algorithm that enumerates non-adaptive partial policies $\pi_1,\pi_2,\dotsc$, each stopping after $a'$ tests, in time $n^{O_\eps(a'/a)}$ among which there is some $\pi_j$ with
	\begin{equation*}
		\Pr[\cost(\pi_j) \ge i]\le \Pr[(1 + \eps) \cost(\pi^\star) \ge i]
	\end{equation*}
	for each $i\in\{a,a+1,\dotsc,a'\}$.
\end{lemma}
We are particularly interested in the following consequence of the lemma.
\begin{corollary}\label{cor:pq}
	Given $\eps > 0$ and $a, a'\in [n]$ with $a < a'$, there is an algorithm that finds in time $n^{O_\eps(a'/a)}$ a non-adaptive partial policy $\pi$ stopping after $a'$ tests and satisfying
	\begin{multline*}
		\sum_{i = a}^{a' - 1} \Pr[\cost(\pi) \ge i]
		+ a' \cdot \Pr[\cost(\pi) \ge a'] \\
		\le \sum_{i = a}^{a'-1} \Pr[(1 + \eps) \cost(\pi^\star) \ge i]
		+ a' \cdot \Pr[(1 + \eps) \cost(\pi^\star) \ge a'].
	\end{multline*}
\end{corollary}
\begin{proof}
	Towards this, apply \Cref{lem:pq} and return the policy that minimizes the left-hand side.
	Note that the left-hand side for each solution can be computed in polynomial time using dynamic programming, where the table we compute contains for every $i=1,2,\dotsc,n$ and every $k'=1,2,\dotsc,i$ the probability of having exactly $k'$ 1s after the $i$th test. An algorithm for computing the entries in this table was given in~\cite[Section 7]{grammel2022}, and we give an algorithm computing the expected cost of a policy for an arbitrary symmetric Boolean function in \Cref{sec:computing}.
	Since $\pi_j$ (as in \Cref{lem:pq}) satisfies the inequality, so does the minimizer.
\end{proof}

Using this we prove the main result.
\begin{proof}[Proof of \Cref{thm:ptasmain} assuming \Cref{lem:pq}]
        Note that
        \begin{equation*}
            \E[\cost(\pi^\star)] \ge \sum_{i=1}^n \Pr\bigg[(1+\eps)\cost(\pi^\star) \in [i, i+1)\bigg] \cdot \frac{i}{1 + \eps} .
        \end{equation*}
        Thus,
        \begin{equation*}
            (1 + \eps) \E[\cost(\pi^\star)] \ge \sum_{i=1}^n \Pr\bigg[(1+\eps)\cost(\pi^\star) \in [i, i+1)\bigg] \cdot i
            = \sum_{i = 1}^n \Pr[(1 + \eps)\cost(\pi^\star) \ge i].
        \end{equation*}
	Let $a_j(i) = 2^{1/\eps\cdot j + i}$ for each $i=0,\dotsc,1/\eps-1$ and $j\in \N_0$.
	Observe that $$\{a_j(0) \mid j\in \N_0\},\dots,\{a_j(1/\eps-1) \mid j\in \N_0\}$$ forms a partition of $\{2^i \mid i\in \N_0\}$. We will show that there exists a partition class that ``contributes'' only little to the cost of the optimum. Towards this, observe
	\begin{multline*}
		\sum_{i=0}^{1/\eps - 1}\sum_{j\in \N_0} a_j(i) \cdot \Pr[(1 + \eps)\cost(\pi^\star) \ge a_j(i)]
		= \sum_{i\in \N_0} 2^i \cdot \Pr[(1 + \eps)\cost(\pi^\star) \ge 2^i] \\
		\le 2 \sum_{i = 1}^n \Pr[(1 + \eps)\cost(\pi^\star) \ge i]
		= 2(1 + \eps) \E[\cost(\pi^\star)] .
	\end{multline*}
	Thus, there exists some $\ell$ that contributes only a small fraction to the left-hand side, i.e.,
	\begin{equation}
        \label{eq:random-shift}
		\sum_{j\in \N_0} a_j(\ell) \cdot \Pr[\cost(\pi^\star) \ge a_j(\ell)]
		\le 2\eps(1 + \eps) \cdot \E[\cost(\pi^\star)] .
	\end{equation}
    In the following we assume that $\ell$ is known to the algorithm. Formally, the algorithm runs for every possible choice of $\ell$, computes the expected cost of the resulting policy (again using dynamic programming~\cite[Section 7]{grammel2022}) and outputs the one with the lowest expected cost. 
    
	For sake of brevity, write $a_j$ instead of $a_j(\ell)$. Further,
	define $a_0 = 1$, and let $h$ be minimum such that $a_{h
    +1}\geq n$.
    For $j\in\{0,\dots,h\}$, let $\pi_j$ be the partial policy generated by applying~\Cref{cor:pq} with $a=a_j$ and $a'=a_{j+1}$.
    We define the final policy $\pi$ as $\pi_0\circ \pi_1\circ\dots\circ \pi_h$.

    For some $j\in\{1,\dotsc,h\}$ we consider how many tests from $\pi_j$ are actually performed in $\pi$. If $\cost(\pi_{j-1}) < a_j$, then none of the tests of $\pi_j$ are performed (except for duplicates appearing in $\pi_0,\dotsc,\pi_{j-1}$), since, by the end of $\pi_{j-1}$, $\pi$  has already determined the function value.
    Otherwise, we may or may not perform tests from $\pi_j$ (depending on $\pi_0,\dots,\pi_{j-2}$), but never more than
    \begin{equation*}
        a_j + \sum_{i=a_j + 1}^{a_{j+1} - 1} \mathbf{1}_{\cost(\pi_j) \ge i} .
    \end{equation*}
	By linearity of expectation, it follows that 
	\begin{align*}
		\E[\cost(\pi)] &\le 
		\sum_{i=1}^{a_{1} - 1} \Pr[\cost(\pi_0) \ge i] \\ 
        &+ \sum_{j=1}^h \left[a_{j} \cdot \Pr[\cost(\pi_{j-1}) \ge a_{j}] + \sum_{i=a_j+1}^{a_{j+1} - 1} \Pr[\cost(\pi_j) \ge i] \right] \\
		&\le \sum_{j=0}^h \left[\sum_{i=a_j}^{a_{j+1} - 1} \Pr[\cost(\pi_j) \ge i] + a_{j+1} \cdot \Pr[\cost(\pi_j) \ge a_{j+1}] \right] \\
		&\le \sum_{j=0}^h \left[\sum_{i=a_j}^{a_{j+1} - 1} \Pr[(1 + \eps)\cost(\pi^\star) \ge i] + a_{j+1} \cdot \Pr[(1 + \eps)\cost(\pi^\star) \ge a_{j+1}]  \right] \\
		&\le (1 + \eps) \E[\cost(\pi^\star)] + \sum_{j=1}^h a_{j} \cdot \Pr[(1 + \eps)\cost(\pi^\star) \ge a_{j}] \\
		&\le (1 + \eps)\E[\cost(\pi^\star)] + 2(1 + \eps)\eps \cdot \E[\cost(\pi^\star)] \le (1 + 4\eps) \cdot\E[\cost(\pi^\star)],
	\end{align*}
    where we use the property guaranteed by~\Cref{cor:pq} in the third step and \Cref{eq:random-shift} in the fourth step.
        Since $a_{j+1} / a_j \le 2^{1/\eps}$ for all $j$, the running time is bounded by $n^{2^{O(1/\eps)}}$.
	Scaling $\eps$ with a factor of $1/4$ reduces the approximation ratio to $1 + \eps$ while preserving the running time above.
\end{proof}

\subsection{Algorithm for the Bounded Variant}
\label{subsec:bounded}

The goal of this subsection is to show \Cref{lem:pq}, which will complete the proof of \Cref{thm:ptasmain}. Our algorithm will pick tests so as to \emph{dominate} certain parts of the optimal solution. The notion of dominance is the following.

Let $V,V^\star \subseteq [n]$ with $|V| \geq |V^\star|$. For $h\in\mathbb{N}$, denote by $[-h] = \{n,n-1,\dots,n-h+1\}$. We say that $V$
\textit{dominates} $V^\star$ (written $V \succeq V^\star$) if, for any $h \in [n]$,
\begin{itemize}
    \item $|V \cap [h]| \geq |V^\star \cap [h]|$ (called \textit{left dominance}) and
    \item $|V \cap [-h]| \geq |V^\star \cap [-h]|$ (called \textit{right dominance}).
\end{itemize}

Equivalently, there exists an injection $\ell: V^\star \rightarrow V$ such that $\ell(v) \leq v$ for all $v \in V^\star$ (left dominance) \textit{and} there exists an injection $r: V^\star \rightarrow V$ such that $r(v) \geq v$ for all $v \in V^\star$ (right dominance). Recall that the variables are sorted by their probabilities, so the injections above satisfy that $p_{\ell(v)}\le p_v$ and $p_{r(v)}\ge p_v$.

Clearly, if $|V| = |V^\star|$, then $V \succeq V^\star$ implies $V = V^\star$. But if $|V|>|V^\star|$, then the sets can be different. For example, if $V^\star$ contains the middle third of $[3n]$ and $V = [3n] \setminus V^\star$, then $V \succeq V^\star$ and yet $V$ and $V^\star$ are disjoint. It turns out that even without full knowledge of $V^\star$, but with an appropriately chosen small fraction of the elements (called \emph{milestones} in the following) of $V^\star$, we can efficiently find a set $V$ which is \emph{guaranteed} to dominate $V^\star$ and does not contain many more elements.

We first show that dominance is a desirable property.

\begin{lemma}
    \label{lem:propstop}
    Let $\pi$ and $\pi'$ be partial non-adaptive policies. Suppose that for some $\ell,\ell' \in [n]$, the length-$\ell'$ prefix of $\pi'$ dominates the length-$\ell$ prefix of $\pi$. Then
    \[
        \Pr[\cost(\pi') > \ell'] \leq \Pr[\cost(\pi) > \ell].
    \]
\end{lemma}

\Cref{lem:propstop} follows from the following lemma in a relatively straightforward way. We defer the
proof of \Cref{lem:propstop} to \Cref{sec:computing}. 

\begin{lemma}
    \label{lem:dominance}
    Let $V,V^\star \subseteq [n]$ be such that $V \succeq V^\star$. 
    Then, for any $\ell \in [|V^\star|]$ we have
    \begin{align*}
        \mathrm{Pr}\left[\sum_{i\in V} x_i \geq \ell\right] \geq \mathrm{Pr}\left[\sum_{i\in V^\star} x_i \geq \ell\right] \text{ and }
        \mathrm{Pr}\left[\sum_{i\in V} (1-x_i) \geq \ell\right] \geq \mathrm{Pr}\left[\sum_{i\in V^\star} (1-x_i) \geq \ell\right].
    \end{align*}
\end{lemma}
\begin{proof}
    We focus on showing the first inequality; the proof for the second inequality is symmetric.
    Let $\ell \in [|V^\star|]$ be arbitrary.
    Since $V \succeq V^\star$, there exists an injective mapping $f: V^\star \rightarrow V$ such that $p_i \leq p_{f(i)}$ for all $i \in V^\star$.
    
    We couple $X_{V^\star} = \{x_i \mid i \in V^\star\}$ and $X_{V} = \{x_i \mid i \in V\}$ by demanding that $x_i=1$ implies $x_{f(i)}=1$ for all $i\in V^\star$. (If $f(i)=i$, this is a vacuous demand.) This is possible since $p_i \leq p_{f(i)}$ for all $i \in V^\star$. Also note that $X_{V^\star}$ is still independent and $X_{V}$ is still independent (but $X_{V^\star}\cup X_{V}$ is not independent, unless $V^\star=V$); hence the inequality that is to be shown remains unaffected.

    For all $i\in V^\star$, now define $\delta_i := x_{f(i)}-x_i$ and notice that, by our coupling, $\delta_i$ is non-negative. Thus,
    $$\mathrm{Pr}\left[\sum_{i\in V} x_i \geq \ell\right]
    \geq \mathrm{Pr}\left[\sum_{i\in V^\star} x_{f(i)} \geq \ell\right]
    = \mathrm{Pr}\left[\sum_{i\in V^\star} x_{i}+\delta_i \geq \ell\right]
    \geq \mathrm{Pr}\left[\sum_{i\in V^\star} x_{i} \geq \ell\right],$$
    where we use injectivity of $f$ in the first step, the definition of $\delta_i$ in the second step, and non-negativity of $\delta_i$ in the third step. The claim follows.
\end{proof}

Recall that we are not only interested in a single inequality of the type that \Cref{lem:propstop} states; \Cref{lem:pq} demands multiple such inequalities. To this end, we do not only seek to dominate a single set. It will, however, be sufficient to think of the optimal solution in terms of a sequence of $b\in O_\eps(a'/a)$ disjoint sets (``buckets'') so that the order within each bucket does not matter. Then, we aim to find another sequence of disjoint sets, also of length $b$, such that we have the aforementioned dominance property for each two corresponding prefixes of the two sequences.

Formally, let $(V^\star_1,V^\star_2,\dots,V^\star_b)$ be a
$b$-tuple of disjoint subsets of $[n]$ with $|V^\star_i|\geq 1/\eps$ for all $i\in[b]$. 
We are going to enumerate a number of
$b$-tuples of the form $(V_1,V_2,\dots,V_{b})$ with the following properties:
\begin{itemize}
    \item[(i)] For all $b$-tuples in the enumeration, $V_1,\dots,V_b$ are disjoint subsets of $[n]$,
    \item[(ii)] For all $b$-tuples in the enumeration, $|V_i| \leq (1+2\eps)|V^\star_i|$ for all $i \in [b]$, and
    \item[(iii)] For at least one $b$-tuple in the enumeration, it holds that $\bigcup_{i'=1}^i V_{i'} \succeq \bigcup_{i'=1}^i V^\star_{i'}$ for all $i \in [b]$.
\end{itemize}

We first show that this will indeed to lead to inequalities akin to~\Cref{lem:pq} (if we require certain sizes of $V^\star_i$ for all $i\in[b]$).

\begin{lemma}\label{lem:ann}
    Let $a$ and $a'$ be positive integers with $2(1+\eps)^{2}/\eps\leq a < a' \leq n$, and let $b$ be another positive integer.
    Furthermore:
    \begin{itemize}
        \item Let $(V^\star_1,V^\star_2,\dots,V^\star_b)$ be a $b$-tuple of disjoint subsets of $[n]$ such that $\pi^\star=\pi^\star_1\circ\dots\circ\pi^\star_{b}\circ\pi^\star_{b+1}$, $V^\star_i=\set(\pi_i^\star)$ for all $i\in[b]$, $|V^\star_1| = \lfloor (a-1)/(1+2\eps)\rfloor$, and $|V^\star_i| \leq \lceil\eps|V^\star_1|\rceil$ for all $i \in [b]\setminus\{1\}$.
        \item Let $(V_1,V_2,\dots,V_b)$ be a $b$-tuple of disjoint subsets of $[n]$ with \[
    \bigcup_{i'=1}^i V_{i'} \succeq \bigcup_{i'=1}^i V^\star_{i'}
    \]
    and $|V_i| \leq (1+2\eps)|V^\star_i|$ for all $i \in [b]$. Also let $\pi=\pi_1\circ\dots\circ \pi_b$ be a partial policy, where $\pi_i$ is an arbitrary partial policy with $\set(\pi_i)=V_i$ for all $i\in[b]$.
    \end{itemize}
    Then
    \[
        \Pr[\cost(\pi) \geq \ell] \leq \Pr[(1+2\eps)^{3}\cost(\pi^\star) \geq \ell]
    \]
    for all $\ell\in\{a,a+1,\dots,a'\}$.
\end{lemma}

\begin{proof}
    Consider some $\ell \in \{a,a+1,\dots,a'\}$.

    Let $i \leq b$ be the largest integer such that $\sum_{i'=1}^i |V_{i'}| \leq \ell-1$ and let $\ell' = \sum_{i'=1}^i |V_{i'}|$.
    We observe that
    \begin{equation}\label{eq:v1-size}
        |V_1| \leq (1+2\eps)|V^\star_1| \leq a-1 \leq \ell - 1,
    \end{equation}
    where the second inequality follows from the choice of $|V^\star_1| = \lfloor (a-1)/(1+2\eps) \rfloor$. Thus, $i \geq 1$.
    Since $|V^\star_i| \leq \lceil\eps|V^\star_1|\rceil$ for $i \geq 2$ and as $i$ is maximal, we have
    \begin{align}
    \ell' &\geq \ell-1 - \lceil \eps|V^\star_1|\rceil \notag\\
          &> \ell-1 - \eps|V^\star_1|-1\notag\\
          &\geq \ell - 1- (a-1) \frac{\eps}{1+2\eps}-1 \notag\\
          &\geq \ell - 1 - (\ell-1) \frac{\eps}{1+2\eps}-1 \notag\\
          &= \frac{\ell-1}{1+2\eps} -1\notag\\
          &> \frac{\ell}{1+2\eps} - 2\notag\\
          &\geq \frac{\ell}{(1+2\eps)^2},\label{eq:l-lprime}
    \end{align}
    where the third inequality follows from \Cref{eq:v1-size} and the last inequality holds as $\ell \geq a \geq 2(1+\eps)^{2}/\eps$.

    Let $\ell^\star = \sum_{i'=1}^i |V^\star_{i'}|$. Using that $\Pr[\cost(\pi') \geq x]$ does not increase as $x$ increases for any partial policy $\pi'$, we obtain
    \begin{align*}
        \Pr[\cost(\pi) \geq \ell] &\leq \Pr[\cost(\pi) > \ell-1]\\
        &\leq \Pr[\cost(\pi) > \ell'] \\
        &\leq \Pr\left[\cost(\pi^\star) > \ell^\star\right]    \\
        &\leq \Pr\left[\cost(\pi^\star) > \ell'/(1+2\eps)\right]\\
        &\leq \Pr\left[\cost(\pi^\star) \geq \ell/(1+2\eps)^{3}\right].
    \end{align*}
    Indeed, the second inequality follows from the definition of $\ell'$. The third inequality follows from $\bigcup_{i'=1}^i V_{i'} \succeq \bigcup_{i'=1}^i V^\star_{i'}$ and \Cref{lem:propstop}. The fourth inequality follows from $\ell^\star \geq \ell'/(1+2\eps)$, which is due to $|V_i| \leq |V^\star_i|/(1+2\eps)$ for all $i \in [b]$. Finally, the last inequality follows from~\Cref{eq:l-lprime}.
\end{proof}

It remains to show that we can indeed enumerate $b$-tuples with the desired properties (i)--(iii). Towards this, consider some $(V^\star_1,V^\star_2,\dots,V^\star_b)$,
$i \in [b]$, and $j \in [1/\eps-1]$. We denote by $m(V^\star_i,j)$ the $(\lfloor
j\eps|V^\star_i|\rfloor)$-th smallest element in $V^\star_i$. We call $m(V^\star_i,j)$ the $j$-th \textit{milestone} of $V^\star_i$. (These milestones will later be ``guessed.'')

We present an algorithm (\Cref{alg:alg}) that receives the sizes $|V^\star_1|,\dots,|V^\star_b|$ as well as such milestones as input and computes, when the milestones are correct, a $b$-tuple $(V_1,V_2,\dots,V_{b})$ with the desired properties (in particular (iii)).

The algorithm does the following for each $i\in[b]$. There is a counter $c$ that is $\eps|V^\star_i|$ initially. The algorithm first does a forward pass over the elements (\emph{the forward loop}, from line~\ref{line:loop1-first} to~\ref{line:loop1-last}), greedily adding available elements to $V_i$ as long as the value of $c$ is at least $1$. During this pass, we increment $c$ by $\eps|V^\star_i|$ whenever we encounter a milestone. The forward loop alone is enough to guarantee left dominance. Note that $c$ takes non-integer values if $\eps|V^\star_i|$ is not integer. 

Then, the algorithm increments the counter by $\lceil \eps|V^\star_i|\rceil$ and starts the \emph{backward loop} (lines~\ref{line:loop2-first} to~\ref{line:loop2-last}), which does a backward pass over the elements, greedily adding available elements to $V_i$ as long as the counter $c$ is at least $1$. This is intuitively what ensures right dominance.

\begin{algorithm}[h]
    \caption{}
    \label{alg:alg}
    \KwIn{$\eps>0$ such that $1/\eps\in\N$; $|V^\star_i|\geq 1/\eps$, and $m(V^\star_i,j)$ for $i \in [b]$ and $j \in [1/\eps-1]$}
    $V_i \gets \emptyset$ for all $i \in [b]$\;\label{line:init}
    \For{$i \in [b]$}{
        $c \gets \eps|V^\star_i|$\;
        \For{$f$ from $1$ to $n$}{\label{line:loop1-first}
            \If{$f \in \{m(V^\star_i,j) \mid j \in [1/\eps-1]\}$}{
                $c \gets c+\eps|V^\star_i|$\;\label{line:loop1-inc}
            }
            \If {$c \geq 1$ and $f \notin \bigcup_{i'=1}^i V_{i'}$\label{line:loop1-check}}{
                $V_i \gets V_i \cup \{f\}$\;\label{line:loop1-add}
                $c \gets c - 1$\;\label{line:loop1-last}
            }
        }
        $c \gets c + \lceil\eps |V^\star_i|\rceil$\;\label{line:loop2-inc}
        \For{$f$ from $n$ to $1$}{\label{line:loop2-first}
            \If {$c \geq 1$ and $f \notin \bigcup_{i'=1}^i V_{i'}$}{\label{line:loop2-check}
                $V_i \gets V_i \cup \{j\}$\;\label{line:loop2-add}
                $c \gets c - 1$\;\label{line:loop2-last}
            }
        }
    }
    \Return $(V_1,V_2,\dots,V_b)$\;
\end{algorithm}

We illustrate this algorithm in \Cref{fig:dominance}. The following lemma shows that \Cref{alg:alg} fulfills its purpose.

\begin{figure}[h]
    \centering
    \usetikzlibrary{patterns,patterns.meta}
\tikzset{
  picked/.style={
    pattern={Lines[angle=45,distance=2pt]},
    pattern color=gray
  }
}
\tikzset{
  milestone/.style={
    pattern={Hatch[angle=45,distance=1pt]},
    pattern color=gray
  }
}
\begin{subfigure}{\textwidth}
\centering
\begin{tikzpicture}[scale=0.5]
    \useasboundingbox (-1,-1) to (21,2);
    \draw (0,0) rectangle node {$1$} (1,1);
    \draw (1,0) rectangle node {$2$} (2,1);
    \draw[picked] (2,0) rectangle node {$3$} (3,1);
    \draw (3,0) rectangle node {$4$} (4,1);
    \draw (4,0) rectangle node {$5$} (5,1);
    \draw[milestone] (5,0) rectangle node {$6$} (6,1);
    \draw (6,0) rectangle node {$7$} (7,1);
    \draw[picked] (7,0) rectangle node {$8$} (8,1);
    \draw (8,0) rectangle node {$9$} (9,1);
    \draw[milestone] (9,0) rectangle node {$10$} (10,1);
    \draw (10,0) rectangle node {$11$} (11,1);
    \draw[picked] (11,0) rectangle node {$12$} (12,1);
    \draw (12,0) rectangle node {$13$} (13,1);
    \draw (13,0) rectangle node {$14$} (14,1);
    \draw[milestone] (14,0) rectangle node {$15$} (15,1);
    \draw[picked] (15,0) rectangle node {$16$} (16,1);
    \draw (16,0) rectangle node {$17$} (17,1);
    \draw (17,0) rectangle node {$18$} (18,1);
    \draw[picked] (18,0) rectangle node {$19$} (19,1);
    \draw (19,0) rectangle node {$20$} (20,1);
\end{tikzpicture}
\subcaption{An example set $V^\star_1$ indicated with shade. Milestones for $\eps = 1/4$ indicated with darker shade.}
\end{subfigure}

\begin{subfigure}{\textwidth}
\centering
\begin{tikzpicture}[scale=0.5]
\useasboundingbox (-1,-1) to (21,2);
        \draw[picked] (0,0) rectangle node {$1$} (1,1);
        \draw[picked] (1,0) rectangle node {$2$} (2,1);
        \draw (2,0) rectangle node {$3$} (3,1);
        \draw (3,0) rectangle node {$4$} (4,1);
        \draw (4,0) rectangle node {$5$} (5,1);
        \draw[milestone] (5,0) rectangle node {$6$} (6,1);
        \draw[picked] (6,0) rectangle node {$7$} (7,1);
        \draw (7,0) rectangle node {$8$} (8,1);
        \draw (8,0) rectangle node {$9$} (9,1);
        \draw[milestone] (9,0) rectangle node {$10$} (10,1);
        \draw[picked] (10,0) rectangle node {$11$} (11,1);
        \draw (11,0) rectangle node {$12$} (12,1);
        \draw(12,0) rectangle node {$13$} (13,1);
        \draw (13,0) rectangle node {$14$} (14,1);
        \draw[milestone] (14,0) rectangle node {$15$} (15,1);
        \draw[picked] (15,0) rectangle node {$16$} (16,1);
        \draw (16,0) rectangle node {$17$} (17,1);
        \draw (17,0) rectangle node {$18$} (18,1);
        \draw[picked] (18,0) rectangle node {$19$} (19,1);
        \draw[picked] (19,0) rectangle node {$20$} (20,1);
\end{tikzpicture}
\subcaption{The set $V_1$ that \Cref{alg:alg} produces given the shaded milestones.}
\end{subfigure}

    \caption{Example of output generated by \Cref{alg:alg} for a single iteration of the main loop.}
    \label{fig:dominance}
\end{figure}

\ \newline

\begin{lemma}\label{lem:alg}
    Let $(V^\star_1,V^\star_2,\dots,V^\star_b)$ be any $b$-tuple of disjoint subsets of $[n]$. The output $(V_1,V_2,\dots,V_b)$ of \Cref{alg:alg} when given $|V^\star_i|\geq 1/\eps$ for $i \in [b]$, and $m(V^\star_i,j)$ for $i \in [b]$ and $j \in [1/\eps-1]$ satisfies \[
    \bigcup_{i'=1}^i V_{i'} \succeq \bigcup_{i' =1}^i V^\star_{i'} 
    \]
    and $|V_i| \leq (1+2\eps)|V^\star_i|$ for all $i \in [b]$.
\end{lemma}

\begin{proof} 
    We start by observing that $|V_i| \leq (1+2\eps)|V_i^\star|$ holds for all $i \in [b]$ as at most \[
    \eps|V_i^\star| + \left(\frac{1}{\eps}-1\right)\eps|V_i^\star| + \lceil\eps|V_i^\star|\rceil \leq (1+\eps)|V_i^\star| + 1 \leq (1+2\eps)|V_i^\star|
    \]
    elements are added to $V_i$, where the last inequality holds as $|V_i| \geq 1/\eps$.

    We now show that $\bigcup_{i'=1}^i V_{i'} \succeq \bigcup_{i' =1}^i V^\star_{i'}$ by induction on $i$. Denote $V_0 = V^\star_0 = \emptyset$ such that the base case $V_0 \succeq V^\star_0$ is trivial. Let $i \geq 1$ and suppose that $\bigcup_{i'=0}^{i-1} V_{i'} \succeq
    \bigcup_{i'=0}^{i-1} V^\star_{i'}$. Let $V = \bigcup_{i'=0}^{i} V_{i'}$ and let $V^\star =
    \bigcup_{i'=0}^{i} V^\star_i$. We must show that $V \succeq V^\star$. We claim the following.

    \bigskip

    \begin{claim}\label{claim:alg}
        For any $h \in [n] \setminus V$, we have 
        \begin{itemize}
        \item[(i)]$|V_i \cap [h]| \geq |V_i^\star \cap [h]|$ and 
        \item[(ii)] $|V_i \cap \{n,n-1,\dots,h\}]| \geq |V_i^\star \cap \{n,n-1,\dots,h\}|$
        \end{itemize}
    \end{claim}
    \claimproof{\Cref{claim:alg}}{  
   
    Consider any $h \in [n] \setminus V$. Let \[
    M = |\{j \in [1/\eps-1] \mid m(V_i^\star,j) < h\}|
    \]
    be the number of milestones smaller than $h$. Note that $h$ itself cannot be a milestone as any milestone is added to $V_i$ on line~\ref{line:loop1-add} after $c$ is incremented by at least $1$ on line~\ref{line:loop1-inc}.

    We start by observing that \begin{equation}
    \label{eq:capbound}
       \lfloor M\eps|V_i^\star|\rfloor \leq |V_i^\star \cap [h-1]| \leq |V_i^\star \cap [h]| \leq \lfloor(M+1)\eps|V_i^\star|\rfloor,
    \end{equation}
    which follows directly from the definition of milestones and the fact that $h$ is not a milestone. See \Cref{fig:claima}.

    \begin{figure}
        \centering
        \begin{tikzpicture}[scale=0.5]
    \node (lab) at (-1.5,0.5) {$V_i$};
    \draw[picked] (0,0) rectangle node {$1$} (1,1);
    \draw[picked] (1,0) rectangle node {$2$} (2,1);
    \draw[] (2,0) rectangle node {$3$} (3,1);
    \draw (3,0) rectangle node {$\cdots$} (9.5,1);
    \draw (9.5,0) rectangle node {$h$} (10.5,1);
    \draw (10.5,0) rectangle node {$\cdots$} (17,1);
    \draw (17,0) rectangle node {$18$} (18,1);
    \draw (18,0) rectangle node {$19$} (19,1);
    \draw[picked] (19,0) rectangle node {$20$} (20,1);

    \draw[decorate,decoration={brace,amplitude=8pt}] (0,1.2) -- node[yshift=0.8cm] {$M$ milestones} (9.5,1.2);

\begin{scope}[yshift=-1cm]
    \draw[->, very thick] (0,0) to node[fill=white] {Forward loop} (9.5,0);
    \node[rotate=-90] (c1) at (10,-0.5) {\small $c < 1$};
    \draw[->, very thick] (10.5,0) to (20,0);
    \draw[->, very thick] (20,-1) to node[fill=white] {Backward loop} (10.5,-1);
\end{scope}

\begin{scope}[yshift=-4cm]
    \node (lab2) at (-1.5,0.5) {$V^\star_i$};
    \draw[picked] (0,0) rectangle node {$1$} (1,1);
    \draw[] (1,0) rectangle node {$2$} (2,1);
    \draw[picked] (2,0) rectangle node {$3$} (3,1);
    \draw (3,0) rectangle node {$\cdots$} (9.5,1);
    \draw[picked] (9.5,0) rectangle node {$h$} (10.5,1);
    \draw (10.5,0) rectangle node {$\cdots$} (17,1);
    \draw (17,0) rectangle node {$18$} (18,1);
    \draw[picked] (18,0) rectangle node {$19$} (19,1);
    \draw (19,0) rectangle node {$20$} (20,1);

    \draw[decorate,decoration={brace,amplitude=8pt,mirror}] (0,-0.2) -- node[yshift=-0.8cm,align=left] {$\leq\lfloor (M+1)\eps |V_i^\star|\rfloor$} (9.5,-0.2);
    \draw[decorate,decoration={brace,amplitude=8pt,mirror}] (10.5,-0.2) -- node[yshift=-0.8cm,align=left] {$\leq |V_i^\star| - \lfloor M\eps |V_i^\star|\rfloor$} (20,-0.2);
\end{scope}
\end{tikzpicture}
        \caption{Illustration of the setup in the proof of \Cref{claim:alg}. Shaded boxes indicate elements in the corresponding set.}
        \label{fig:claima}
    \end{figure}
    
    To show (i), consider the iteration of the forward loop where $f = h$. As $h \notin V$, $h$ was not added to $V_i$ on line~\ref{line:loop1-add}, implying that $c < 1$ on line~\ref{line:loop1-check}. As $c$ is initially $\eps|V_i^\star|$ and was incremented by $M\eps|V_i^\star|$ before this iteration, we must have \[
    |V_i \cap [h]| = \lfloor (M+1)\eps|V_i^\star|\rfloor \geq |V_i^\star \cap [h]|,
    \]
    where the last inequality follows from \Cref{eq:capbound}.

    Define $h' = n-h+1$ such that $[-h'] = \{n,n-1,\dots,h\}$.
    To show (ii), consider the iteration of the backward loop where $f=h$. As $h \notin V$, $h$ was not added to $V_i$ on line~\ref{line:loop2-add}, implying that $c < 1$ in this iteration. In fact, we must have $c=0$ as the fractional increments of $c$ sum to $\eps|V_i^\star| + (1/\eps-1)\eps|V_i^\star| = |V_i^\star|$ and $c$ is only ever decremented by $1$. Since $c$ was incremented by $(1/\eps-1-M)\eps|V_i^\star|+\lceil\eps|V_i^\star|\rceil$ between the iteration of the forward loop where $f=h$ and the iteration of the backward loop where $f=h$, we must have \[
    |V_i \cap [-h']| \geq (1/\eps-1-M)\eps|V_i^\star| + \lceil\eps|V_i^\star|\rceil
    \]
    We see that also
    \begin{equation*}
    |V_i^\star \cap [-h']| = |V_i^\star| - |V_i^\star \cap [h-1]| \leq |V_i^\star| - \lfloor M\eps|V_i^\star|\rfloor
    \end{equation*}
    where we use \Cref{eq:capbound} in the second inequality. Subtracting the above two inequalities we get
    \begin{align*}
    |V_i^\star \cap [-h']|-|V_i \cap [-h']| &\leq  |V_i^\star| - \lfloor M\eps|V_i^\star|\rfloor - (1/\eps-1-M)\eps|V_i^\star| - \lceil\eps|V_i^\star|\rceil \\
    &= (M+1)\eps|V_i^\star| - \lceil\eps|V_i^\star|\rceil - \lfloor M\eps|V_i^\star|\rfloor \\
    &< 1,
    \end{align*}
    from which it follows that $|V_i^\star \cap [-h']|-|V_i \cap [-h']| \leq 0$ as both terms are integers.
    }

    We return to the proof of the lemma. 
    Suppose that $|V \cap [h]| < |V^* \cap [h]|$ for some $h \in [n]$ and let $h$ be minimal with this property. By the minimality of $h$, we must have $h \notin V\supseteq V_i$. By the induction hypothesis and \Cref{claim:alg}, \begin{align*}
    |V \cap [h]| &= |(V \setminus V_i) \cap [h]| + |V_i \cap [h]| \\
                     &\geq |(V^\star \setminus V^\star_i) \cap [h]| + |V^\star_i \cap [h]| \\
                     &= |V^\star \cap [h]|,
    \end{align*}
    which is a contradiction.

    Suppose that $|V \cap [-h']| < |V^* \cap [-h']|$ for some $h' \in [n]$ and let $h'$ be minimal with this property. Let $\ma{h}=n-h'+1$. Then $h \notin V$. By the induction hypothesis and \Cref{claim:alg}, \begin{align*}
    |V \cap [-h']| &= |(V \setminus V_i) \cap [-h']| + |V_i \cap [-h']| \\
                     &\geq |(V^\star \setminus V^\star_i) \cap [-h']| + |V^\star_i \cap [-h']| \\
                     &= |V^\star \cap [-h']|,
    \end{align*}
    which is again a contradiction.
\end{proof}

\Cref{lem:pq} now follows relatively directly by combining the last two lemmas with full enumeration.

\begin{proof}[Proof of~\Cref{lem:pq}] We need to distinguish a few cases, to cover cases in which we cannot apply \Cref{lem:alg} and \Cref{lem:ann}:
    \begin{itemize}[leftmargin=2cm] 
        \item[Case 1:] $a<4(1+\eps)/\eps^2+1$. In this case we can in time $n^{O_\eps(a'/a)}$ fully enumerate all partial policies of length $a'$, showing the claim, even without the $1+\eps$ factor.
        \item[Case 2:] $a\geq 4(1+\eps)/\eps^2+1$. This implies that $a\geq 2(1+\eps)^{2}/\eps$ (using $1/\eps\in\mathbb{N}$), which is needed to apply \Cref{lem:ann}. The same assumption allows us to define $|V_1^\star|=\lfloor (a-1)/(1+2\eps)\rfloor$ and get that $\eps|V_1^\star|\geq 2/\eps$.

        \begin{itemize} 
            \item[Case 2a:] $a'-\lfloor (a-1)/(1+2\eps)\rfloor < 1/\eps$. We will run into trouble defining bucket sizes so as to apply \Cref{lem:alg}. Since $a$ may be too large, we cannot simply enumerate all partial policies of length $a'$. We define $b=2$ and $|V_2^\star|=a'-|V_1^\star|$. We apply full enumeration to obtain $m(V^\star_1,j)$ for all $j \in [1/\eps-1]$. We obtain $V_1$ by \Cref{alg:alg} and $V_2$ by full enumeration, in total time $n^{O_\eps(1)}$. Then, for the correct elements from the enumerations, \Cref{lem:alg} guarantees that also the dominance condition needed to apply \Cref{lem:ann} is fulfilled.
            \item[Case 2b:] $a'-\lfloor (a-1)/(1+2\eps)\rfloor\geq 1/\eps$. We define $|V_2^\star|,\dots,|V_b^\star|$ in the following way: Start with a counter with value equal to the total size of $|V_2^\star|,\dots,|V_b^\star|$, which, by definition of $|V^\star_1|$, is $a'-\lfloor (a-1)/(1+2\eps)\rfloor\geq 1/\eps$. Open a new bucket of size $\lceil\eps|V_1^\star|/2\rceil \geq 1/\eps$ and decrease the counter by $\lceil\eps|V_1^\star|/2\rceil$ until the counter
            drops below $\lceil\eps|V_1^\star|/2\rceil$. Increase the size of the final bucket by the remaining counter value. This way, $1/\eps \leq |V_i^\star|\leq 2\lceil\eps|V_1^\star|/2\rceil-1 \leq \eps|V_1^\star|$ for all $i\in[b]$, and $b\in O_\eps(a'/a)$. We may thus enumerate $m(V^\star_i,j)$ for all $i\in[b]$ and $j \in [1/\eps-1]$, in time $n^{O_\eps(a'/a)}$, and apply \Cref{lem:alg} to each element in the enumeration. For the correct element from the enumeration, we may then also apply \Cref{lem:ann}.
        \end{itemize}
    \end{itemize}
    Note that, in Cases 2a and 2b, we need to scale down $\eps$ by a constant to obtain the necessary guarantee from the application of \Cref{lem:ann}.
\end{proof}

\subsection{The Expected Cost of Evaluating Symmetric Functions}
\label{sec:computing}

In this section, we give a useful expression for the expected cost of a
non-adaptive policy for any symmetric Boolean function. We use this
expression to give a polynomial-time algorithm for computing this expected
cost and to prove \Cref{lem:propstop}.

Fix a symmetric Boolean function $f : \{0,1\}^n \to \{0,1\}$. Recall that it is given through thresholds $0 = t_1 < t_2 <
\dots < t_T < t_{T+1} = n+1$ such that the function values changes at the thresholds, i.e., $f(x) \neq f(y)$ if $t_j \leq
\sum_{i=1}^n x_i < t_{j+1}$ and $t_{j+1} \leq
\sum_{i=1}^n y_i < t_{j+2}$, for $j \in [T-1]$. We keep $f$ and these
thresholds fixed throughout this section. 

\begin{lemma}
    \label{lem:costsymmetric}
    Let $\pi$ be any non-adaptive policy for $f$, let $\ell \in
    \{0,1,\dots,n\}$, and let $B \subseteq [T]$ be the set of indices $j \in
    [T]$ such that $\ell \geq t_j + (n - t_{j+1} + 1)$. Then
    \[
        \Pr[\cost(\pi) \leq \ell] = \sum_{j\in B}
        \left(1-\left(\Pr\left[\sum_{i=1}^{\ell} x_{\pi(i)} < t_j\right] +
            \Pr\left[\sum_{i=1}^{\ell} (1-x_{\pi(i)}) <
            n-t_{j+1}+1\right]\right)\right).
        \]
\end{lemma}

\begin{proof}
    We start by observing that the value of $f$ is determined if, for some $j
    \in [T]$, at least $t_j$ $1$s have been found \text{and} enough $0$s
    have been found to rule out that the total number of $1$s is $t_{j+1}$
    or more. Thus, $\pi$ is done evaluating $f$ after $\ell$ tests iff 
    \begin{equation}
        \label{eq:stopcondition}\exists j \in [T] : \sum_{i=1}^\ell x_{\pi(i)} \geq t_j \text{ and
        } \sum_{i=1}^\ell (1-x_{\pi(i)}) \geq n-t_{j+1}+1.
    \end{equation}
    Note that \eqref{eq:stopcondition} cannot be satisfied simultaneously
    for distinct $j \in [T]$. Furthermore, \eqref{eq:stopcondition} cannot
    be satisfied at all for $j \notin B$ as adding the two inequalities gives
    \[
        \ell = \sum_{i=1}^\ell x_{\pi(i)} + \sum_{i=1}^\ell (1-x_{\pi(i)}) \geq
        t_j + (n - t_{j+1} + 1).
    \]
    Now, we have
    \begin{align*}
        \Pr[\cost(\pi) \leq \ell] 
        &= \sum_{j = 1}^T \Pr\left[\sum_{i=1}^{\ell} x_{\pi(i)} \geq t_j \text{ and } \sum_{i=1}^{\ell} (1-x_{\pi(i)}) \geq n-t_{j+1} + 1 \right] \\ 
        &= \sum_{j \in B} \Pr\left[\sum_{i=1}^{\ell} x_{\pi(i)} \geq t_j \text{ and } \sum_{i=1}^{\ell} (1-x_{\pi(i)}) \geq n-t_{j+1} + 1 \right] \\ 
        &= \sum_{j\in B}\left(1 - \Pr\left[ \sum_{i=1}^{\ell} x_{\pi(i)} < t_j \text{ or
    } \sum_{i=1}^{\ell} (1-x_{\pi(i)}) < n-t_{j+1}+1\right]\right) \\
    &= \sum_{j\in B} \left(1-\left(\Pr\left[\sum_{i=1}^{\ell} x_{\pi(i)} < t_j\right] +
        \Pr\left[\sum_{i=1}^{\ell} (1-x_{\pi(i)}) < n-t_{j+1}+1\right]\right)\right),
\end{align*}
where the first equality follows since \eqref{eq:stopcondition} cannot be
satisfied for distinct $j \in [T]$ simultaneously, the second since
\eqref{eq:stopcondition} cannot be satisfied for $j \notin B$. The last
inequality follows from the fact that the two events in the disjunction are
disjoint as adding the two inequalities gives $\ell < t_j + (n - t_{j+1} +
1)$, which does not hold for $j \in B$.
\end{proof}

We now give a polynomial-time algorithm for computing the
expected cost of a non-adaptive policy for any symmetric Boolean function.
This was done for $k$-of-$n$ functions in \cite{grammel2022}, and our
approach is an extension of this using \Cref{lem:costsymmetric}.

\begin{thm}
    \label{thm:computecost}
    The expected cost $\E[\cost(\pi)]$ of any non-adaptive policy $\pi$
    for any symmetric Boolean function $f$ can be computed in polynomial
    time.
\end{thm}

\begin{proof}
Fix a non-adaptive policy $\pi$. 
For $k,\ell \in \{0,1,\dots,n\}$, define \[ 
    d(\ell,k) = \Pr\left[\sum_{i=1}^\ell x_{\pi(i)} = k\right]. 
\]

The values $d(\ell,k)$ can be computed in time $O(n^2)$ using a
straightforward dynamic program with the recurrence
\[
    d(\ell,k) = d(\ell-1,k)(1-p_{\pi(i)}) +d(\ell-1,k-1)p_{\pi(i)}.
\]

%
%

Using the values $d(\ell,k)$, we can compute the values
\[
    \Pr\left[\sum_{i=1}^\ell x_{\pi(i)} < t_j\right] = \sum_{k=0}^{t_j-1} d(\ell,k)
    \text{ and } \Pr\left[\sum_{i=1}^\ell (1-x_{\pi(i)}) < n-t_{j+1}+1\right] =
    \sum_{k=0}^{n-t_{j+1}} d(\ell,\ell-k),
\]
which is all we need to evaluate the right-hand side of the expression in
\Cref{lem:costsymmetric} for all $\ell \in \{0,1,\dots,n\}$.

Using $\Pr[\cost(\pi) \leq \ell]$ for all $\ell \in \{0,1,\dots,n\}$, we
can compute $\E[\cost(\pi)]$ using
\[
    \E[\cost(\pi)] = \sum_{\ell=1}^n \Pr[\cost(\pi) \geq \ell] =
    \sum_{\ell=1}^n 1 - \Pr[\cost(\pi) \leq \ell-1].
\]
This completes the proof.
\end{proof}

We now move on to proving \Cref{lem:propstop}. The proof again relies on
\Cref{lem:costsymmetric}.

\begin{proof}[Proof of \Cref{lem:propstop}]
    Fix policies $\pi$ and $\pi'$ and lengths $\ell$ and $\ell'$ such that
    the length-$\ell'$ prefix of $\pi'$ dominates the length-$\ell$ prefix
    of $\pi$. 

    Let $B \subseteq [T]$ be the sets of indices $j \in [T]$ such that
    $\ell \geq t_j + (n - t_{j+1} + 1)$ and similarly let $B' \subseteq
    [T]$ be the sets of indices $j \in [T]$ such that $\ell' \geq t_j + (n
    - t_{j+1} + 1)$. As the length-$\ell'$ prefix of $\pi'$ dominates the
    length-$\ell$ prefix of $\pi$, we have $\ell' \geq \ell$ and thus $B
    \subseteq B'$.

    By \Cref{lem:costsymmetric}, since $B \subseteq B'$, and the
    dominance of the $\ell'$-length prefix of $\pi'$ over the $\ell$-length
    prefix of $\pi$, we have
    \begin{align*}
        \Pr[\cost(\pi') \leq \ell'] 
        &= \sum_{j \in B'} \left(1 - \left(\Pr\left[ \sum_{i=1}^{\ell'} x_{\pi'(i)} < t_j\right] + \Pr\left[\sum_{i=1}^{\ell'} (1-x_{\pi'(i)}) < n-t_{j+1}+1\right]\right)\right) \\
        &\geq \sum_{j \in B} \left(1 - \left(\Pr\left[ \sum_{i=1}^{\ell'}
                x_{\pi'(i)} < t_j\right] + \Pr\left[\sum_{i=1}^{\ell'}
        (1-x_{\pi'(i)}) < n-t_{j+1}+1\right]\right)\right) \\
        &\geq \sum_{j \in B} \left(1 - \left(\Pr\left[ \sum_{i=1}^{\ell}
                x_{\pi(i)} < t_j\right] + \Pr\left[\sum_{i=1}^{\ell}
            (1-x_{\pi(i)}) < n-t_{j+1}+1\right]\right)\right) \\
        &= \Pr[\cost(\pi) \leq \ell].
    \end{align*}
    We now get \[
        \Pr[\cost(\pi') > \ell'] = 1-\Pr[\cost(\pi') \leq \ell'] \leq
        1-\Pr[\cost(\pi) \leq \ell] = \Pr[\cost(\pi) > \ell], 
    \] which completes the proof.
\end{proof}

\section{Conclusion}

First, it remains as an open question whether there also exists a PTAS for the arbitrary-cost case. 
It seems plausible that one can use standard techniques to get the number of relevant cost classes per bucket down to a logarithmic number and to then use a similar approach as described in \Cref{subsec:bounded} for each cost class, which would result in a QPTAS. It is unclear to us whether this approach can be modified to obtain a PTAS. 

Second, recall that it is not known whether computing the optimal non-adaptive policy for evaluating symmetric functions is NP-hard, even in the arbitrary-cost case. Although there are many problems in stochastic combinatorial optimization with this status, and it is common to develop approximation algorithms for them, that fact may be particularly intriguing for this fundamental problem.

\bibliography{jour}

\end{document}